\documentclass[11pt,letterpaper]{article}

\usepackage{amsmath,amsfonts,graphpap,amscd,mathrsfs,graphicx,lscape}
\usepackage{epsfig,amssymb,amstext,xspace}
\usepackage{theorem, bbm}
\usepackage{algorithm,algorithmic}

\usepackage{color}              % Need the color package
\usepackage{epsfig}
\usepackage{authblk}

\usepackage{color-edits}

\addauthor{rpl}{red}
\addauthor{mb}{blue}
\addauthor{bl}{green}

\newcommand{\E}{\mathbb{E}}

\newtheorem{theorem}{Theorem}[section]
\newtheorem{lemma}[theorem]{Lemma}

\newtheorem{corollary}[theorem]{Corollary}

\newtheorem{observation}[theorem]{Observation}

\newtheorem{example}[theorem]{Example}

{\theorembodyfont{\rmfamily} }
\def\squareforqed{\hbox{\rule{2.5mm}{2.5mm}}}

\def\QED{\ifmmode\squareforqed % in mathmode : print just the square
  \else{\nobreak\hfil   % \hfil to end of current line
    \penalty50                 % penalty 50 for breaking the line here
    \hskip1em                  % leave at least 1em before the square
    \null                      % \hbox{}
    \nobreak                   % prohibit line break
    \hfil                      % another \hfil (if a break occurred)
    \squareforqed              % put the square here
    \parfillskip=0pt           % the line really ends here
    \finalhyphendemerits=0     % ignore a hyphen on the line above
    \endgraf}                  % end the paragraph
  \fi}

\def\blksquare{\rule{2mm}{2mm}}
\def\qedsymbol{\blksquare}
\newcommand{\bg}[1]{\medskip\noindent{\bf #1}}
\newcommand{\ed}{{\hfill\qedsymbol}\medskip}

\newenvironment{proof}{\bg{Proof : }}{\ed}
\newenvironment{proofof}[1]{\bg{Proof of #1 : }}{\ed}

\newcommand{\R}{\ensuremath{\mathbb R}}
\newcommand{\Z}{\ensuremath{\mathbb Z}}

\newcommand{\vc}[1]{\mathbf{#1}}

\newcommand{\comment}[1]{}
\newcommand{\fullversion}[2]{#2}{}  %Alternative to be used for 10 page
\newcommand{\junk}[1]{}

\newlength{\tmp} \newlength{\lpsx} \newlength{\lpsy} \newlength{\upsx} \newlength{\upsy}

\newcommand{\one}{\mathbf{1}}

\newcommand{\nash}{\textsc{Nash}}
\newcommand{\bnash}{\textsc{BNash}}

\newcommand{\add}{\textsc{Add}}
\newcommand{\gs}{\textsc{Gs}}
\newcommand{\ud}{\textsc{Ud}}
\newcommand{\sm}{\textsc{Sm}}
\newcommand{\xos}{\textsc{Xos}}

\begin{document}

\setcounter{page}{0}

\title{On the Efficiency of the Walrasian Mechanism}


%\numberofauthors{3} %  in this sample file, there are a *total*

\author[1]{Moshe Babaioff}
\author[1]{Brendan Lucier}
\author[1]{Noam Nisan}
\author[1]{Renato Paes Leme}

\affil[1]{Microsoft Research, \texttt{\{moshe, brlucier,
noamn, renatop\}@microsoft.com}}

\renewcommand\Authands{ and }

\date{}

\maketitle

\begin{abstract}
 Central results in economics guarantee the existence of efficient equilibria for various classes of markets. An underlying assumption in early work is that agents are price-takers, i.e., agents honestly report their true demand in response to prices. A line of research in economics, initiated by Hurwicz (1972), is devoted to understanding how such markets perform when agents are strategic about their demands.  This is captured by the \emph{Walrasian Mechanism} that proceeds by collecting reported demands, finding clearing prices in the \emph{reported} market via an ascending price t\^{a}tonnement procedure, and returns the resulting allocation.
 Similar mechanisms are used, for example, in the daily opening of the New York Stock Exchange and the call market for copper and gold in London.

In practice, it is commonly observed that agents in such markets reduce their demand leading to behaviors resembling bargaining and to inefficient outcomes. We ask how inefficient the equilibria can be. Our main result is that the welfare of every pure Nash equilibrium of the Walrasian mechanism is at least one quarter of the optimal welfare, when players have gross substitute valuations and do not overbid. Previous analysis of the Walrasian mechanism have resorted to large market assumptions to show convergence to efficiency in the limit. Our result shows that approximate efficiency is guaranteed regardless of the size of the market.

We  extend our results in several directions. First, our results extend to
Bayes-Nash equilibria and outcomes of no regret learning via the smooth
mechanism framework.
We also extend our bounds to any mechanism that maximizes welfare with respect to the declared valuations and never charges agents more than their bids.  Additionally, we consider other classes of valuations and bid spaces beyond those satisfying the gross substitutes conditions.
Finally, we relax the no-overbidding assumption, and present bounds that are parameterized by the extent to which agents are willing to overbid.

\end{abstract}

\renewcommand{\thepage}{}
\clearpage
\pagenumbering{arabic}

\section{Introduction}

%Studying how prices get adjusted in markets is a central problem in economic theory.
The manner in which market prices are set and adjusted is a central area of study in economic theory.
A formal approach to this topic was proposed by Walras \cite{walras}, who defined the concept
of competitive (aka Walrasian) equilibrium:
an assignment of prices to goods such that, when each agent takes his preferred allocation under
the given price vector, the market clears (i.e., all goods are sold) and no good is over-demanded.
%allocation and prices to agents such that each agent takes his
%favorite bundle of items under the current prices and the market clears, i.e., all goods
%are sold.
An important property of the Walrasian equilibrium, known as the First Welfare
Theorem, states that whenever an equilibrium exists, the allocation is efficient. Following
Walras' original work, the existence of such competitive equilibria for various types
of economies has been shown \cite{KelsoCrawford, GulStachetti}, often accompanied by simple and distributed algorithmic
procedures to compute such prices \cite{CheungColeDevanur, nisansegal, Murota03}.

The Walrasian equilibrium is suggestive of a process by which prices are adjusted in markets,
but it is also often used directly as a mechanism to allocate goods. Double auctions, which are prevalent
in finance, essentially work by computing a price that clear the market and executing as many
trades as possible at that price.  The opening price of the New York Stock Exchange is computed
in such a way. The way the prices of copper and gold in London are adjusted
follows a similar procedure: the demands of agents are elicited and prices are computed
according to such equilibria. We refer to Rustichini, Satterthwaite and Williams \cite{Satterthwaite94}
for a more extensive discussion.

It has been commonly observed that traders might strategically reduce their demands for certain
goods seeking more favourable prices. As Rustichini et al \cite{Satterthwaite94} write: \emph{``Such behavior,
which is the essence of bargaining, may lead to an impasse that delays or lessens the gains from trade.''}
This leads to a natural question: to what extent does strategic behavior of economic agents, rather than
price-taking behavior, hurt the efficiency of a market?

Hurwicz \cite{Hurwicz1972} proposed a game-theoretical framework to analyze settings where agents
are strategic. In such model each economic agent is a player in a non-cooperative game and their
strategy is a report of their preferences. After reporting their preferences, which serves as a proxy
for their demands, the competitive equilibrium in the \emph{declared market} determines allocation
and payments to the agents, who evaluate their outcome with respect to their \emph{true preferences}.
We will call this the \emph{Walrasian Game} or \emph{Walrasian Mechanism}. Hurwicz \cite{Hurwicz1972}
observes that truthfully reporting preferences is not always an equilibrium of the Walrasian game. An initial characterization
of equilibria was done by Hurwicz and then extended by Otani and Sicilian \cite{Otani_Sicilian82, Otani_Sicilian90},
who showed that inefficient outcomes arise as the equilibria of the Walrasian mechanism.

The previously mentioned work focuses on economies with divisible goods. For indivisible goods,
Gul and Stacchetti \cite{GulStachetti, GulStachetti00} show that truthful revelation of demands
is \emph{not} in general an equilibrium for the Walrasian Mechanism, except in the special cases
of additive valuations (in which it corresponds to isolated English auctions) and in the
unit-demand case (which corresponds to the ascending auction of Demange, Gale and Sotomayor
\cite{DemangeGaleSotomayor}). For generic gross substitute valuations, no dynamic ascending
price auction can be truthful \cite{GulStachetti00}.

A natural response to such observations is to seek conditions under which the equilibria of
the Walrasian game resemble the competitive equilibrium one would obtain from the truthful
reports. An intuition first formalized by Roberts and Postlewaite \cite{Roberts_Postlewaite}
is that in large markets, the ability of each individual player to influence the market is
minimal, so agents should approximately behave as price-taking agents. They capture the
concept of large markets through \emph{replica economies}, i.e., they consider the
equilibria of a game where there are $k$ identical copies of each agent and each good and
study the limit as $k$ goes to infinity. The original result of Roberts and Postlewaite
shows that the incentives for the agents not to act as price-takers are vanishingly
small in the size of the economy. Jackson and Manelli \cite{Jackson_Manelli} show
that under some regularity conditions, the equilibrium allocations in large markets
will be close to the allocations in the competitive equilibrium. A version of this
result for double auctions was later provided by Rustichini, Satterthwaite and Williams
\cite{Satterthwaite94} and Satterthwaite and Williams \cite{Satterthwaite02}
for a simple market with unit-demand traders. Recently, Azevedo and
Budish \cite{azevedo_budish} proposed the notion of \emph{strategyproof in the large}
that generalize this idea to other strategic settings such as matching markets.

\paragraph{Our results}
In this paper we seek to provide efficiency guarantees for the equilibria of
market-clearing mechanisms, such as the Walrasian game, without resorting
to large market assumptions.\footnote{{Our bounds apply even though, in small markets, the agents' behaviour does not necessarily resemble price-taking.}}
Since it is known that Nash equilibria of the Walrasian
mechanism may be inefficient (Otani and Scicilian \cite{Otani_Sicilian82}), we aim
to show that \emph{all} equilibria are \emph{approximately} efficient, i.e., the ratio
between the welfare of an optimal allocation and the welfare of a Nash equilibrium is
bounded. This ratio is referred to as the \emph{Price of Anarchy} of the game. Our main
result is a bound on this ratio, which \comment{.
Our bound depends only on one parameter, the \emph{exposure factor}, which is a bound on the risk tolerance of the
agents (described in more detail below).
It is interesting to note that,} unlike previous results,  is
independent on number of players, items or on any distributional assumptions on the valuations.

We follow the model of Hurwicz \cite{Hurwicz1972} where the strategy of each player is
a reported valuation over the goods in the markets.  We model goods as indivisible and
heterogeneous items and assume that players have combinatorial valuations over the goods.
Our model of economy follows Gul and Stachetti \cite{GulStachetti, GulStachetti00} :
we assume that each player is initially endowed with a sufficiently large amount of
money that neither budgets nor initial endowments influence the equilibrium outcome.

We begin by focusing specifically on the Walrasian mechanism.
Recall that this mechanism restricts
%By restricting
the agents to report only valuations satisfying the \emph{gross substitutes} property, which guarantees the existence of a  Walrasian equilibrium in the declared market (Gul and Stachetti \cite{GulStachetti}) and thereby ensures that the mechanism is well-defined.
% We note that, in order to guarantee the existence of Walrasian equilibrium, the Walrasian mechanism restricts agents to report only valuations satisfying the \emph{gross substitutes} property.
%Assuming gross subtitutability is also necessary to ensure the existence of a Walrasian
%equilibrium (Gul and Stachetti \cite{GulStachetti}) and therefore is necessary to ensure that the
%Walrasian game is well-defined.
Our first result is that if the agents' true valuations
satisfy the gross substitutes property, {then for any Nash equilibrium in which no agent bids
in a way that exposes himself to the possibility of obtaining negative utility ex post\footnote{{A sufficient condition for this property is that no agent declares more than his true value on any set of goods.}},}
%\mbedit{under the assumption that agents never bid in a way that expose them to the ''risk'' of ending up with negative utility}\rpldelete{then in any Nash equilibrium in which no agent
%bids more than his true value on any set of goods},
the equilibrium outcome generates
welfare that is at least one fourth of the socially efficient outcome.
Moreover, there always exists at least one equilibrium satisfying the required property.
This result extends to mixed Nash and coarse correlated equilibria, as well as
{the Bayesian setting (incomplete information).}
% MB: the Beyseian result need more than jus beliefs, it also need that types actually are distributed according the priors.
In particular, we show that the expected efficiency
of any Bayes-Nash equilibrium of the Walrasian mechanism is at least one fourth of the
expected optimal welfare.  This extension to incomplete information settings is
done via the smooth mechanisms framework of Syrgkanis and Tardos \cite{Syrgkanis13},
which is based on Roughgarden's notion of smooth games \cite{Roughgarden09, Roughgarden12}.

The above results can be extended along multiple dimensions.
{First, we can relax the requirement that agents do not engage in ``risky'' bidding
behavior, as follows.}
We parameterize the loss in
efficiency by the \emph{exposure factor}, which is a measure of
risk tolerance.
We say that a strategy has exposure factor $\gamma$ if, by playing
such strategy, an agent is guaranteed to never pay
more than $(1+\gamma)$ times his value for the items he receives, under any declarations of the
other agents.
{We show that, for the Walrasian mechanism, the ratio between the socially optimal welfare and the equilibrium welfare
will be at most $(4+2\gamma)$ at any equilibrium in which each agent's strategy has exposure factor at most $\gamma$.}

Additionally, our efficiency guarantees extend to circumstances in which the
true preferences of the agents do not satisfy the gross substitutes condition,
but instead are drawn from more general classes of valuation functions, such as
submodular or fractionally subadditive valuations.  That is, even in circumstances
where a Walrasian equilibrium does not exist with respect to the true preferences
of the agents, a mechanism that restricts the participants to \emph{report} gross
substitutes valuations and then allocates according to a pricing equilibrium of the
reported valuations will achieve approximately welfare-maximizing outcomes at
equilibrium.  Indeed, we show that this result holds even if the declared preferences are
restricted even further to the class of additive valuations.

Finally, our efficiency guarantees can be generalized to apply more broadly than the
Walrasian mechanism.  Indeed, our results apply to any mechanism
that chooses an allocation maximizing the welfare of the \emph{declared valuations}
and charges prices no larger then the bids. It is trivial to see that the Walrasian
mechanism has these properties. Under this extension, our techniques can be applied to
other mechanisms such as the VCG mechanism and the "pay-your-bid" mechanism (which
generates the welfare-optimal allocation but charges payments equal to the declared
values for goods received).

The application to the VCG mechanism is notable, as it relates to a recent line of
work on efficiency guarantees in auctions with reduced expressiveness.  For example,
the VCG mechanism with reports restricted to be additive functions is precisely the
``simultaneous single-item auction" studied first by Christodoulou, Kov{\'a}cs and
Schapira \cite{CKS08}, then subsequently in a line of research aimed at analyzing
the efficiency of this mechanism under various classes of agent
valuations \cite{BR11,HKMN11,PLST12,syrgkanis,ST12,Feldman13}.  Our main result
implies a price of anarchy bound on this auction format, under a wide variety of
payment rules, when agents' true valuations are fractionally subadditive.
For this specific case of the VCG auction where agents are constrained to use a bidding
language more restrictive than the valuation space, price of anarchy bounds were 
independently obtained by D\"{u}tting, Henzinger and Starnberger \cite{DHS13}.

%\rplcomment{Removed the comments on Bayesian and mixed, since we already talked about them
%earlier. I am focusing this paragraph on correlated equilibria.}
In the case that agents are able to express their true valuations (i.e., the type space
and bidding language are the same), our bounds on the efficiency of Bayes-Nash equilibria
hold even if the distribution from which values are drawn exhibit correlations among agents.
This is in contrast with item bidding auctions, which have good efficiency if
valuations of the agents are drawn from independent distributions, but are known to perform
very poorly in settings where the valuations of the players are correlated
\cite{KshipraThesis, Feldman13}. Underlying this fact is the phenomenon known as the
\emph{exposure problem} -- which refers to the fact in item bidding auctions, that agents can't
fully express their valuations without exposing themselves to the risk of negative utility.
Consider, for example, the extreme case of an unit-demand player, i.e., a player that wants
at most one item. He has to choose between bidding high on multiple items and exposing himself
to the risk of winning and paying for more then one item, or rather place a ``safe" bid and
be severely constrained on how he can express his preferences. Our results highlight
that the ability of expressing your true value is crucial in order to get equilibria with
good efficiency in settings where valuations are drawn from correlated distributions.

\paragraph{Other related work} Other recent papers have studied Nash equilibria of games
induced by market mechanisms. Adsul et al \cite{adsul10} study the Nash equilibria
of the game induced by the Linear Fisher Market, showing existence of equilibrium and providing
a complete polyhedral characterization in some special cases. Chen, Deng, Zhang \cite{chen11} and
Chen, Deng, Zhang, Zhang \cite{Chen12} study \emph{incentive ratios} in Fisher Markets, i.e., they
bound how much the utility of any given player can improve by strategic play in comparison to his
utility if he were to play truthfully. Markakis and Telelis \cite{Markakis12} and
de Keijzer et al \cite{Markakis13} study the Price of Anarchy of Uniform Price Auctions, which
can be cast as a game derived from a market equilibrium computation.

\section{Preliminaries}

\paragraph{Notation} Throughout the paper we will denote vectors by bold
letters: $\vc{p} = (p_1, \hdots, p_m) \in \R^m$ will denote a vector of prices over $m$ items.
Given a subset $S \subseteq [m]$, we will denote by $\one_S$ the indicator
vector of set $S$. When $S = [m]$, we will omit the subscript, i.e., $\one =
\one_{[m]}$ is the vector where all components are $1$. Similarly $\vc{0} =
\one_{\emptyset}$ is the vector of all zeros. Given two vectors $\vc{x},
\vc{y}$ we will denote by $\vc{z} = \vc{x}\cup\vc{y}$ and $\vc{w} = \vc{x}\cap \vc{y}$ by the vectors such that $z_i = \max\{x_i,y_i\}$ and $w_i = \min\{x_i, y_i\}$ respectively. Also, for binary vectors $\vc{x} \in \{0,1\}^m$ we say $j \in \vc{x}$ if $x_j = 1$. We will also denote dot products as follows: $\vc{x} \cdot \vc{y}
= \sum_i x_i y_i$.

Valuation functions will be represented by  $v: \{0,1\}^m \rightarrow \R_+$
such that $v(\vc{0}) = 0$ and $v(\vc{x}) \leq v(\vc{y})$ whenever $\vc{x} \leq \vc{y}$.  Given a valuation function $v$, we define the demand at price vector $\vc{p}$ as $D_v(\vc{p}) = \text{argmax}_{\vc{x} \in \{0,1\}^m} [ v(\vc{x}) - \vc{p}\cdot \vc{x} ]$. Whenever convenient, we will see $v$ as a function $v:\Z_+^m \rightarrow \R_+$ such that $v(\vc{x}) = v(\vc{x} \cap \one)$. Given any function $f:\Z_+^m \rightarrow \R_+$, we define marginal values as $f(\vc{y} \vert \vc{x}) = f(\vc{y} + \vc{x}) - f(\vc{x})$.

\paragraph{Classes of valuation functions} A class of valuation functions is a
subset of $\{v; v:\{0,1\}^m \rightarrow \R_+\}$. In what follows, we will
consider various subsets of increasing level of complexity.
\begin{itemize}
\item Additive valuations: $v \in \add$ if $v(\vc{x}) = \vc{w} \cdot \vc{x}$
for some $\vc{w} \in \R^m_+$.
\item Unit demand valuations: $v \in \ud$ if $v(\vc{x}) = \max_{j \in \vc{x}}
w_j$ for some $\vc{w} \in \R^m_+$.
\item Gross substitutes valuations: $v \in \gs$ if for any pair of prices
$\vc{p} \leq \vc{q}$, $S = \{j; p_j = q_j\}$ and $\vc{x} \in D_v(\vc{p})$ there
is $\vc{y} \in D_v(\vc{q})$ such that $\one_S \leq \vc{y}$.
\item Submodular valuations: $v \in \sm$ if for any $\vc{x}, \vc{y} \in
\{0,1\}^m$, $v(\vc{x}) + v(\vc{y}) \geq v(\vc{x}\cup \vc{y}) + v(\vc{x}\cap
\vc{y})$. An equivalent definition is that for all $\vc{x} \cap \vc{y} = \vc{0}$ and $\vc{z} \leq \vc{y}$ then $v(\vc{x} \vert \vc{z}) \geq v(\vc{x} \vert \vc{y})$.
\item XOS valuations: $v \in \xos$ if there is a set $I$ and a set of vectors
$\{\vc{w}_i\}_{i \in I}$ such that $v(\vc{x}) = \max_{i \in I} \vc{w}_i \cdot
\vc{x}$.
\end{itemize}

It is known that $(\add \cup \ud) \subsetneq \gs \subsetneq \sm \subsetneq \xos$. We
refer to Lehmann, Lehmann and Nisan \cite{LehmannLehmannNisan} for a more
extensive discussion on the relation between such classes. One important
property of those classes is the closure with respect to the OR ($\vee$)
operator. Given a set of valuations $\{v_i\}$, we define $\vee_i v_i : \Z_+^m \rightarrow
\R_+$ as: $(\vee_i v_i)(\vc{x}) = \max\{ \sum_i v_i(\vc{x}_i); \sum_i \vc{x}_i \leq \vc{x}, \vc{x}_i \leq \one \}$. In other words, the value of the OR of multiple agents' valution functions evaluated at $\vc{x}$ is the value of
the optimal partition of the goods in $\vc{x}$ among the agents. Notice that
the function is defined over the domain of $\vc{x} \in \Z^m_+$ so to allow the optimal
partition that allocated item $j$ at most $x_j$ times.

We say that a valuation class is closed under the OR operator when given any
valuations $v_1, \hdots, v_n$ in the class, the valuation $\vee_i v_i$ restricted to
$\{0,1\}^m$ is also in the class. It is known from
\cite{LehmannLehmannNisan,Murota96} that the classes $\add$, $\gs$ and $\xos$
are closed under the OR operator.

\paragraph{Market economy} Consider a market with $n$ agents and $m$ goods,
each agent with a valuation $v_i : \{0,1\}^m \rightarrow \R_+$ over the set of
goods and quasi-linear utilities, i.e., the utility of agent $i$ to be
allocated a bundle $\vc{x}_i \in \{0,1\}^m$ for a total payment of $\pi_i$ is
$u_i = v_i(\vc{x}_i) - \pi_i$. A \emph{Walrasian equilibrium} in such market is
a vector of prices $\vc{p} \in \R^m_+$ and a partition of the goods into
disjoint bundles $\one = \sum_i \vc{x}_i$ such that for each player $i$,
$\vc{x}_i \in D_{v_i}(\vc{p})$. The First Welfare Theorem states that whenever
a Walrasian equilibrium exists, it maximizes welfare, i.e., the partition
maximizes $\sum_i v_i(\vc{x}_i)$. We call a vector $\vc{p} \in \R^m_+$
\emph{Walrasian prices} if there is an allocation that paired with such vector
forms a Walrasian equilibrium. The Second Welfare Theorem states that given any
partition of the items $\vc{y}_i$ maximizing $\sum_i v_i(\vc{y}_i)$ and any
vector of Walrasian prices $\vc{p}$, the pair composed of this vector and those
allocations is a Walrasian equilibrium.

A classical result due to Kelso and Crawford \cite{KelsoCrawford} guarantees
the existence of Walrasian equilibria if the valuation functions are
\emph{gross substitutes}. Gul and Stachetti \cite{GulStachetti} show that this
condition is in some sense necessary: gross substitutes is the largest class of
valuation functions containing unit-demand valuations for which a Walrasian
equilibrium is always guaranteed to exist.

Gul and Stachetti \cite{GulStachetti} also show that the set of Walrasian
prices forms a \emph{lattice}, i.e., for any valuations $v_1, \hdots, v_n \in
\gs$, if $\vc{p}$ and $\vc{p'}$ are Walrasian prices for such valuations, then
$\vc{p}\cap \vc{p'}$ and $\vc{p}\cup \vc{p'}$ are also Walrasian prices. This
implies in particular that there exist Walrasian price vectors
$\vc{\underline{p}}$ and $\vc{\overline{p}}$ such that for all Walrasian prices
$\vc{p}$, it is the case that $\vc{\underline{p}} \leq \vc{p} \leq
\vc{\overline{p}}$. The existence proof in Kelso and Crawford
\cite{KelsoCrawford} is constructive and yields a simple and natural ascending
price procedure called Walrasian t\^{a}tonnement that computes a Walrasian
equilibrium. Later in \cite{GulStachetti00}, Gul and Stachetti argue that
this procedure produces the equilibrium corresponding to the lowest point in
the lattice, i.e., with prices $\vc{\underline{p}}$.

Both price vectors $\vc{\underline{p}}$ and $\vc{\overline{p}}$ have a clean
description in terms of the \emph{welfare function}.
%\note{BJL: What is the correct citation for this?  Is it Gul and Stachetti?}
Associate with $\vc{v} =
(v_1, \hdots, v_n)$, the function $W^{\vc{v}} : \Z^m_+ \rightarrow \R_+$ given
by $W^{\vc{v}}(\vc{x}) = (\vee_i v_i)(\vc{x})$. 
{Gul and Stachetti \cite{GulStachetti} show that} $\vc{\underline{p}}$ and
$\vc{\overline{p}}$ can be calculated by the following closed-form formulas:
$$\underline{p}_j = W^{\vc{v}}(\one_j \vert \one) \qquad \overline{p}_j =
W^{\vc{v}}(\one_j \vert \one - \one_j).$$
In other words, the price $\underline{p}_j$ is the extra benefit for society
for an \emph{additional} copy of item $j$. The price $\overline{p}_j$ is how
much harm to the welfare of the society removing item $j$ will cause.

\paragraph{Auction games} We want to consider the market economy as a strategic
game, following the model proposed by Hurwicz \cite{Hurwicz1972}. Before we do
that, we define a generic auction game and propose how to study its equilibria.
The setting is composed of $m$ items and $n$ agents with valuations $v_i$ in a
certain valuation space $\mathcal{V}$. A game for such setting consists of a
bidding space $\mathcal{B}$ and allocation and payment functions:
$$\vc{x}_i: \mathcal{B}^n \rightarrow \{0,1\}^m \qquad \pi_i: \mathcal{B}^n
\rightarrow \R_+ $$
The allocation is supposed to be such that $\sum_i \vc{x}_i(\vc{b}) \leq \one$
for all $\vc{b} = (b_1, \hdots, b_n) \in \mathcal{B}^n$. For this paper, we
will be interested in games such that the bidding language is a subset of the
valuation space, i.e., $\mathcal{B} \subseteq \mathcal{V}$. The utilities in
such game are given by $u_i(v_i;\vc{b}) = v_i(\vc{x}_i(\vc{b})) -
\pi_i(\vc{b})$.

As an example, consider the case of the second price (i.e., Vickrey) auction
for a single item.  This is an auction game in which $m=1$, and
where $\mathcal{V} = \mathcal{B} = \R_+$.  The allocation rule is given by
$x_i = 1$ if $i$ has the highest
bid (breaking ties lexicographically) and zero otherwise.  The payment rule
is given by $\pi_i = \max_{j \neq i} b_j$ whenever $x_i = 1$, and zero otherwise.

We will study auction games in two different settings, non-Bayesian and
Bayesian. We also refer to the non-Bayesian setting as the full information setting.

\paragraph{Nash equilibria and Price of Anarchy} We are interested in studying
the Nash equilibria of such games, i.e., for $
\vc{v} \in \mathcal{V}^n$:
$$\nash(\vc{v}) = \{ \vc{b} \in \mathcal{B}^n ; u_i(v_i; \vc{b}) \geq u_i(v_i;
b'_i, \vc{b}_{-i} ), \forall i \in [n], \forall b'_i \in \mathcal{B} \}$$
In particular we are interested in measuring the social welfare in equilibrium
$W^{\vc{v}}(\vc{x}(\vc{b})) = \sum_i v_i(\vc{x}_i(\vc{b}))$ against the optimal
welfare across all partitions of the items.  The maximum and minimum of such
ratio are known as the Price of Anarchy (PoA) and Price of Stability (PoS):
$$\textsc{PoA} = \max_{\vc{v}} \max_{\vc{b} \in \nash(\vc{v})} \frac{\sum_i
v_i(\vc{x}_i^*(\vc{v}))}{\sum_i v_i(\vc{x}_i(\vc{b}))} \qquad \textsc{PoS} =
\max_{\vc{v}} \min_{\vc{b} \in \nash(\vc{v})} \frac{\sum_i
v_i(\vc{x}_i^*(\vc{v}))}{\sum_i v_i(\vc{x}_i(\vc{b}))}$$
where $\vc{x}_i^*(\vc{v})$ is the allocation maximizing $\sum_i
v_i(\vc{x}_i^*(\vc{v}))$.  Note that if $\nash(\vc{v}) = \emptyset$, both the
price of anarchy and the price of stability are defined to be $1$.

\paragraph{Bayesian Equilibria and Bayesian Price of Anarchy} If the valuation
space in endowed with a distribution $\mathbf{D}$ over $\mathcal{V}^n$ (which
need not be independent between agents), one can
study auction games as Bayesian games. For such games the strategy of each
player is a mapping $b_i : \mathcal{V} \rightarrow \mathcal{B}$ and the set of
Bayes Nash equilibria are given by:
$$\bnash(\mathbf{D}) = \{ (b_i:\mathcal{V} \rightarrow \mathcal{B})_{i=1..n} ;
\E_{\mathbf{D}}[u_i(v_i; \vc{b}(\vc{v})) \vert v_i] \geq
\E_{\mathbf{D}}[u_i(v_i; b'_i, \vc{b}_{-i}(\vc{v}_{-i}) ) \vert v_i], \forall i
\in [n], \forall b'_i \in \mathcal{B} \}$$
and the corresponding notions of Bayesian Price of Anarchy and Bayesian  Price
of Stability are given by:
$$\textsc{BPoA} = \max_{\mathbf{D}} \max_{\vc{b} \in \bnash(\mathbf{D})}
\frac{\E_{\mathbf{D}}[\sum_i v_i(\vc{x}_i^*(\vc{v}))]}{\E_{\mathbf{D}} [\sum_i
v_i(\vc{x}_i(\vc{b}(\vc{v})))]} \qquad \textsc{BPoS} = \max_{\mathbf{D}}
\min_{\vc{b} \in \bnash(\mathbf{D})} \frac{\E_{\mathbf{D}} [\sum_i
v_i(\vc{x}_i^*(\vc{v}))]}{\E_{\mathbf{D}} \sum_i v_i(\vc{x}_i(\vc{b}(\vc{v})))]}$$

\paragraph{Exposure factor} Even for the single-item second price auction in
the full information setting (and more generally for the VCG mechanism), it is
not possible to give any reasonable bound for PoA due to the so called
\emph{bullying equilibria}. Consider a setting with a single items and two
agents with values $v_1 = 1$ and $v_2 = \epsilon$ for this item. The bids $b_1
= 0$ and $b_2 = 10$ form a Nash equilibrium with welfare $\epsilon$ while the
optimal welfare is $1$.
% MB: changes below:
Although this is an equilibrium, it is based on an aggressive bid by agent $2$ which exposes him to loss;
agent $2$ can end up with negative utility
if agent $1$ was to change his bid to $1$. In order to get around the issue of bullying that is based on large exposure,
we define what we call the \emph{exposure factor}, which quantifies the amount of
risk an agent with type $v_i$ is exposing himself to by bidding $b_i$. We say that
the strategy $b_i$ has exposure factor $\gamma$ if:
$$\pi_i(b_i, \vc{b}_{-i}) \leq (1+\gamma) \cdot v_i(\vc{x}_i(b_i,
\vc{b}_{-i})), \forall \vc{b}_{-i} \in \mathcal{B}^n$$
We will call a strategy with $\gamma=0$ a \emph{non-exposure} strategy, since
for \emph{any} bids of the other agents, agent $i$ is guaranteed to have
non-negative utility.
Therefore, we will be interested in bounding the Price of
Anarchy across equilibria with $\gamma$ exposure, i.e.:
$$\nash_\gamma(\vc{v}) = \{ \vc{b} \in \nash(\vc{v}); b_i \text{ has } \gamma
\text{ exposure factor for all } i\}$$
We define $\textsc{PoA}_\gamma, \textsc{PoS}_\gamma$ by simply substituting
$\nash$ by $\nash_\gamma$ in the definition. Analogously, we can define
$\bnash_\gamma$ as $\vc{b} \in \bnash(\mathbf{D})$ such that $\E_{v_i}
\pi_i(b_i(v_i), \vc{b}_{-i}(v_i)) \leq (1+\gamma) \cdot
\E_{v_i} v_i(\vc{x}_i(b_i(v_i), \vc{b}_{-i}(v_i)))$ for any function
$\vc{b}_{-i}: \mathcal{V} \rightarrow \mathcal{B}^{n-1}$. Analogously, we can
also define $\textsc{BPoA}_\gamma$ and $\textsc{BPoS}_\gamma$.
As before, each of the above measures of the price of anarchy and the
price of stability are taken to be $1$ whenever the corresponding set
of equilibria is empty.

It is worth noting that for the case of single item second price auction, a common way to
get around \emph{bullying equilibria} is to note that it is a weakly dominated
strategy for an agent to bid above his valuation. This, however, is not necessarily
true for other auction formats (like item bidding) where agents bid on multiple items. 
Correspondingly, much of the prior literature on auction
games for multiple items \cite{BR11, CKS08, Feldman13, Syrgkanis13}  has
imposed some form of a ``bounded overbidding'' assumption on players' strategies
at equilibrium.

\paragraph{Walrasian mechanism} We define a Walrasian mechanism as a game where
we ask each agent to report a valuation $b_i \in \mathcal{B}$ function and
compute a Walrasian equilibrium of the \emph{reported market}. In other words,
$\vc{x}_i: \mathcal{B}^n \rightarrow \{0,1\}^m$ and $\pi_i: \mathcal{B}^n
\rightarrow \R_+$ is a \emph{Walrasian mechanism} if there is a \emph{price
function} $\vc{p}: \mathcal{B}^n \rightarrow \R^m_+$ such that: (i)
$\pi_i(\vc{b}) = \vc{p}(\vc{b}) \cdot \vc{x}_i(\vc{b})$ and (ii)
$(\vc{x}(\vc{b}), \vc{p}(\vc{b}))$ is a Walrasian equilibrium of the market
defined by $\vc{b} = (b_1, \hdots, b_n)$.

It is known that, in general, Walrasian equilibria might not exist. However, if $\mathcal{B} \subseteq \gs$,
then a Walrasian equilibrium is guaranteed to exist. Moreover, because Walrasian prices form a lattice,
we will be interested in two distinct flavors of the Walrasian mechanism for gross substitutes:

\begin{itemize}
\item \emph{English Walrasian Mechanism}: This is the mechanism that implements
the lower point of the lattice of Walrasian prices. Formally, it allocates
according to the optimal allocation with respect to declared values $\vc{b}$
and charges $\pi_i(\vc{b}) = \sum_{j \in \vc{x}_i} W^{\vc{b}}(\one_j \vert
\one)$. This mechanism is also called \emph{English Auctions with
Differentiated Commodities} by Gul and Stacchetti \cite{GulStachetti00}. For
the special case of one item, it corresponds to the second price / English
auction.
\item \emph{Dutch Walrasian Mechanism}: This is the mechanism that implements
the higher point of the lattice of Walrasian prices. Formally, it allocates
according to the optimal allocation with respect to declared values $\vc{b}$
and charges $\pi_i(\vc{b}) = \sum_{j \in \vc{x}_i} W^{\vc{b}}(\one_j \vert \one
- \one_j)$. For the special case of one item, it corresponds to the first price
/ Dutch auction.
\end{itemize}

Also, for the special case where $\mathcal{B} = \add$, these mechanisms
are equivalent to the second- and first-price item bidding auctions, respectively,
as studied in \cite{BR11, CKS08, Feldman13, Syrgkanis13}.

{We note that, like the simultaneous item bidding auctions described above, 
strategies that includes overbidding (i.e., declaring more than 
one's true value for certain sets of goods) are not necessarily weakly dominated
in the Walrasian mechanism.
An example is given in Appendix \ref{appendix:overbidding}.}

\paragraph{Declared Welfare Maximizers} A general class of auction games that
includes the Walrasian mechanism is called the \emph{declared welfare
maximizers}. We say that a mechanism is a declared welfare maximizer if
$\vc{x}_1(\vc{b}), \hdots, \vc{x}_n(\vc{b})$ is a partition of the set of items
maximizing $\sum_i b_i(\vc{x}_i)$. In order for the $b_i$ to have the semantics
of maximum willingness to pay for a bundle, we enforce the following two
properties over the payment function:
(i) $\pi_i(\vc{b}) \leq b_i(\vc{x}_i(\vc{b}))$, i.e., no agent pays for a bundle
more then his declared value for this bundle, and (ii) for any $\vc{b}$ there is
$\vc{b}'_{-i}$ such that $\vc{x}_i(\vc{b}) = \vc{x}_i(b_i, \vc{b}'_{-i})$ and
$\pi_i(b_i, \vc{b}'_{-i}) = b_i(\vc{x}_i(\vc{b}))$, i.e., for any bid on a
subset, there is a set of bids of the other players such that player $i$ will pay
his bid on his allocated subset.

It is simple to see that the Walrasian mechanism is a declared welfare
maximizer. We also consider the following other declared welfare maximizers:
\begin{itemize}
\item \emph{VCG Mechanism}: The mechanism allocates according to the optimal
declared allocation and charges agents according to the externality they impose
on the other players. Formally, we have: $\pi_i = W^{\vc{b}_{-i}}(\vc{x}_i
\vert \one - \vc{x}_{i})$.
\item \emph{Pay-Your-Bid Mechanism}: The mechanism allocates according to the
optimal declared allocation and charges the bids, i.e., $\pi_i = b_i(\vc{x}_i)$.
\end{itemize}

The following observation about declared welfare maximizers follows immediately from the definitions:

\begin{observation}[Bounded Overbidding]\label{obs:overbidding}
Consider a declared welfare maximizer mechanism and agents with valuations
$v_1, \hdots, v_n$. If $b_i$ is a strategy with exposure factor $\gamma$, then:
$b_i(\vc{x}_i(\vc{b})) \leq (1+\gamma) v_i(\vc{x}_i(\vc{b}))$. In the Bayesian
setting, exposure factor $\gamma$ translates to $\E_{\mathbf{D}}
[b_i(\vc{x}_i(\vc{b}))] \leq (1+\gamma) \E_{\mathbf{D}} [v_i(\vc{x}_i(\vc{b}))]$
\end{observation}

\section{Existence of Efficient and Inefficient Equilibria}

In this section we first argue that in the non-Bayesian setting, both the
English and Dutch versions of the Walrasian equilibria have Nash equilibria
that are efficient and only use non-exposure strategies, implying that the Price of Stability is 1. After that we show
that they also often have inefficient equilibria as well. This is also true for the VCG mechanism.

\begin{lemma}
If $\vc{v}$ is a valuation profile consisting of gross substitute valuations,
the English Walrasian mechanism has pure and efficient Nash equilibria that
employ only non-exposure strategies.
\end{lemma}

\begin{proof}
% Since $\vc{v}$ is a valuation profile consisting of gross substitute valuations a Walrasian equilibrium exist for  the market economy defined by $\vc{v}$.
Let $\vc{p} \in \R^m_+$ be a vector of Walrasian prices for the market economy
defined by $\vc{v}$ and let $(\vc{x}_1^*, \hdots, \vc{x}_n^*)$ be the
corresponding optimal allocation. Such prices exist since the valuations are gross substitute. 
We first prove the claim assuming that $p_j > 0$ for all
$j \in [m]$. Consider the following bids $b_i(\vc{y}) = \vc{p} \cdot
(\vc{x}_i^* \cap \vc{y})$. In other words, each agent submits an additive bid
in which he bids on each item he wins in the optimal allocation exactly its
Walrasian price. Clearly the mechanism allocates $\vc{x}_i^*$ to player $i$,
which is an efficient allocation. Now, we need to show that this strategy is
non-exposure and that this is a Nash equilibrium.

\emph{Non-exposure}: Since the payment of the Walrasian mechanism for a bundle
is at most the bid on this bundle, it is enough to show that $b_i(\vc{y}) \leq
v_i(\vc{y})$ on any $\vc{y} \in \{0,1\}^m$. Fix any $\vc{y} \in \{0,1\}^m$ and
let $\vc{z} = \vc{x}_i^* \cap \vc{y}$. Now, $b_i(\vc{y}) = \vc{p} \cdot \vc{z}$
by the definition of $b_i$. Also since $\vc{x}_i^*$ is the set maximizing
player $i$'s utility under prices $\vc{p}$ we know that: $v_i(\vc{x}_i^*) -
\vc{p} \cdot \vc{x}_i^* \geq v_i(\vc{x}_i^* - \vc{z}) - \vc{p} \cdot (\vc{x}_i^*
- \vc{z})$, therefore: $$v_i(\vc{y}) \geq v_i(\vc{z}) \geq v_i(\vc{x}_i^*) -
v_i(\vc{x}_i^* - \vc{z}) \geq \vc{p}\cdot \vc{z} = b_i(\vc{y})$$

\emph{Nash equilibrium}: Under bids $\vc{b}$ the payment of the English
Walrasian mechanism are zero, since no two players bid on the same item.
Therefore $u_i(v_i; \vc{b}) = v_i(\vc{x}_i^*)$. Now, consider a deviation of
player $i$ to some bid $b'_i$ and let $\vc{y} = \vc{x}_i(b'_i, \vc{b}_{-i})$.
By the definition of prices in the English Walrasian auction, the price of each
item in $\vc{y}$ but not in $\vc{x}_i^*$ is at least $p_j$, so $u_i(v_i; b'_i,
\vc{b}_{-i}) \leq v_i(\vc{y}) - \vc{p} \cdot (\vc{y} - \vc{y} \cap \vc{x}_i^*)
\leq v_i(\vc{y}) - \vc{p} \cdot \vc{y} + \vc{p} \cdot \vc{x}_i^*$. Since
$v_i(\vc{x}_i^*) - \vc{p} \cdot \vc{x_i^*} \geq v_i(\vc{y}) - \vc{p} \cdot
\vc{y}$, we have that: $u_i(v_i; b'_i, \vc{b}_{-i}) \leq v_i(\vc{x}_i^*) =
u_i(v_i; \vc{b})$.

The assumption that $p_j > 0$ made in the proof can be easily removed by
slightly changing bids $b_i$ to bid an infinitesimally small amount $\epsilon$ on each item in $\vc{x}_i^*$ that has price zero and for which $i$ has positive marginal value.
\end{proof}

\begin{corollary} The Price of Stability of the English Walrasian auction is
$1$.
\end{corollary}

The same arguments can be made about the Dutch Walrasian auction, yet one needs
to be careful how to deal with tie breaking rules. The adaptation from the
previous proof from the English to the Dutch Walrasian mechanism follows
exactly the same arguments in Hassidim et al \cite{HKMN11} for proving the
existence of efficient equilibria in first-price item bidding auctions with
gross substitute valuations.

Now, we show two examples of inefficient equilibria of the Walrasian mechanism.
The first example highlights the incentives to perform demand reduction, i.e.,
to declare less value than he actually has. This is in line with the
observation of Rustichini et al \cite{Satterthwaite94} that demand reduction is
a common practice in bargaining.

\begin{example}[Incentives to reduce demand]\label{example:first_bound}
Consider a market with two agents and two items. The first agent is unit demand
with value $1+\epsilon$ per item and the second agent is additive with value
$2$ per item. Formally $v_1(\vc{x}) = (1+\epsilon) \max\{x_1, x_2\}$ and
$v_2(\vc{x}) = 2(x_1 + x_2)$.

If agents report truthfully in the English Walrasian mechanism, items are both priced at $1+\epsilon$ and agent $2$ acquires both, getting utility
$2-2\epsilon$. Agent $2$, however, can improve his utility by changing his bid to $b'_2(\vc{x}) = 2 x_1$, without exposing himself. Now, he only acquires the first item, but Walrasian prices are zero, giving him utility $u'_2 = 2 > u_2$.

This produces a Nash equilibrium of welfare $3$, while the optimal allocation
has welfare $4$, showing that $\textsc{PoA}_0 \geq 4/3$ for English Walrasian
Mechanism.
\end{example}

\begin{example}[Inefficiency due to
Miscoordination]\label{example:2_bound_unit_demand}
Now, we consider a $2$ lower bound on the Price of Anarchy due to
miscoordination. Consider two items $A,B$ and two unit demand agents $1,2$ with
valuations $v_1(\vc{x}) = \max\{(2-\epsilon) x_A, x_B\}$ and $v_2(\vc{x}) =
\max\{x_A, (2-\epsilon) x_B\}$. Consider now the following equilibrium in which
both agents miscoordinate and bid a high amount on their least favorite item
and zero on their preferred item: $b_1(\vc{x}) = x_B$ and $b_2(\vc{x}) = x_A$.
Both items get priced at zero under such declarations, agent $1$ is allocated
item $B$ and agent $2$ is allocated item $A$. This is a Nash equilibrium since
the utility of each agent is $1$ and by deviating his value can increase by at
most $1-\epsilon$ but if this happens, his payment will increase by at least
$1$. The welfare in equilibrium is $2$, while the optimal allocation has
welfare $4-2\epsilon$. This shows that $\textsc{PoA}_0 \geq 2$.

For $\gamma > 0$, consider agents with valuations $v_1(\vc{x}) = \max\{(2-\epsilon) x_A, \frac{2}{2+\gamma} x_B\}$ and $v_2(\vc{x}) =
\max\{\frac{2}{2+\gamma} x_A, (2-\epsilon) x_B\}$. The bids $b_1(\vc{x}) = \frac{2(1+\gamma)}{2+\gamma}x_B$ and $b_2(\vc{x}) = \frac{2(1+\gamma)}{2+\gamma} x_A$ form an inefficient equilibrium with exposure factor $\gamma$ and welfare $\frac{4}{2+\gamma}$ while the optimal welfare is $4-2\epsilon$. This implies that $\textsc{PoA}_{\gamma} \geq 2+\gamma$.
\end{example}

It is interesting to notice that in the previous example all valuations are
unit-demand. In such case, the outcome of the English Walrasian mechanism
coincides with the VCG mechanism \cite{GulStachetti00}, in which truthtelling
is a dominant strategy. The example show that the truth being dominant, there
are other equilibria as well that generate inefficient outcomes. We will see
that for the case of the VCG auction, the bound above is essentially tight.

\section{Bounding the inefficiency of equilibria}

Given Examples \ref{example:first_bound} and \ref{example:2_bound_unit_demand},
it is natural to ask how large the gap between the welfare of the optimal
allocation and the worse welfare of a Nash equilibrium of the Walrasian
mechanism can be made. Our main result is an upper bound on such ratio that
holds for any \emph{declared welfare maximizer}, in particular, any flavor of the Walrasian mechanism, the VCG mechanism and the first price mechanism.
This bounds depends only on the exposure factor and holds even the Bayesian setting (even with correlated
distributions) and \emph{doesn't} depend on number of players, number of items or any characteristic of the distribution. Notice in particular the following theorems doesn't make any sort of large market assumptions.

First start by presenting the statement of our results, starting with its version for gross substitute valuations and then discussing extensions to the more general class of $\xos$ valuation. And a specialization of this result for the VCG mechanism. Proofs are left for the following subsections.

\begin{theorem}\label{thm:gs_poa}
If $\mathcal{V} = \mathcal{B} = \gs$ and the mechanism is a \emph{declared
welfare maximizer}, i.e., allocates according to the optimal partition of items with respect to the bids, then $\textsc{PoA}_\gamma \leq 4+2\gamma$.
\end{theorem}

The following theorem is a strict generalization of the previous for Bayesian settings. Theorem \ref{thm:gs_poa} can be recovered as a special case of its Bayesian counterpart Theorem \ref{thm:gs_bpoa} by taking the distribution concentrated on a single valuation profile. The following proof follows from smoothness arguments \cite{Roughgarden09,Roughgarden12,Syrgkanis13,LPL11} and therefore generalizes also to other equilibrium concepts such as mixed Nash, coarse correlated equilibria and outcomes of no-regret learning dynamics.

\begin{theorem}\label{thm:gs_bpoa}
If $\mathcal{V} = \mathcal{B} = \gs$ and $\mathcal{V}$ is endowed with a (possibly correlated) probability distribution $\mathbf{D}$, then $\textsc{BPoA}_\gamma \leq 4+2\gamma$.
\end{theorem}

Next, we extend Theorem \ref{thm:gs_bpoa} to allow the larger class of $\xos$ valuations (that contains, in particular, all submodular valuations).  This extension comes at the cost of a slightly weaker bound on the price of anarchy.

\begin{theorem}\label{thm:xos_poa}
If $\mathcal{V} = \mathcal{B} = \xos$ and the mechanism is a \emph{declared
welfare maximizer} then $\textsc{BPoA}_\gamma \leq 6+4\gamma$.
%If $\mathcal{V}$ is endowed with a distribution, the same bound also holds for Bayes-Nash equilibria, i.e., $\textsc{BPoA}_\gamma$.
\end{theorem}

For the special case of the VCG mechanism, we further improve the bound,
matching the lower bound in Example \ref{example:2_bound_unit_demand}:

\begin{theorem}\label{thm:vcg_poa}
If $\mathcal{V} = \mathcal{B} = \gs$, then for the VCG mechanism,
$\textsc{BPoA}_\gamma \leq 2+\gamma$. If $\mathcal{V} = \mathcal{B} = \xos$, then for the VCG mechanism, $\textsc{BPoA}_\gamma \leq 3+2\gamma$.
%In both cases, the bounds can be extended to $\textsc{BPoA}_\gamma$.
\end{theorem}

{As in Theorem \ref{thm:gs_bpoa}, the above results follow from smoothness
arguments, and therefore generalize to mixed Nash equilibria, coarse correlated
equilibria, and outcomes of no-regret learning dynamics.  Moreover, since pure
Nash equilibria are a special case of Bayes-Nash equilibria, the bounds on
$\textsc{BPoA}_\gamma$ also apply to $\textsc{PoA}_\gamma$.}

In the previous results, we assumed that all valuations $v \in \mathcal{V}$ could be represented in the bidding language $\mathcal{B}$. It is often useful to restrict the bidding language for various reasons:
\begin{itemize}
\item \emph{representation and communication}: very expressive combinatorial valuations like $\xos$ require many bits to be expressed. It might be desirable to restrict to a simpler class, as additive, unit demand or a simple combination of those, for which the valuation can be more simply represented and communicated.
\item \emph{computational efficiency}: there are computationally efficient algorithms to find the optimal allocation when bids are gross substitutes \cite{FujishigeTamura} but it is computationally hard to compute the optimal allocation for general submodular and $\xos$ valuations
\item \emph{simplicity}: a simpler bidding language might lead to simpler and more intuitive design to agents. A prime example of such approach are \emph{item bidding auctions}.
\end{itemize}

We show that the results for Nash equilibria presented above still hold for mechanisms with restricted bidding languages, as long as agents are able to express additive valuations.  These bounds apply to Bayes-Nash equilibria as well, but we require that agent types be independent; that is, the distribution $\mathbf{D}$ is such that for $i \neq j$, $v_i$ and $v_j$ are independent random variables.

\begin{theorem}[Restricted bidding languages]\label{thm:bidding_language}
Consider a declared welfare maximizer mechanism with valuation space $\mathcal{V}$ and bidding language $\mathcal{B}$. Then:
\begin{itemize}
\item if $\mathcal{V} \subseteq \xos$ and $\add \subseteq \mathcal{B} \subseteq \xos$, then $\textsc{PoA}_{\gamma} \leq 6 + 4 \gamma$.
\item if $\mathcal{V} \subseteq \xos$ and $\add \subseteq \mathcal{B} \subseteq \gs$, then $\textsc{PoA}_{\gamma} \leq 4 + 2 \gamma$.
\end{itemize}
For the special case of the VCG mechanism, we can strengthen the bounds:
\begin{itemize}
\item if $\mathcal{V} \subseteq \xos$ and $\add \subseteq \mathcal{B} \subseteq \xos$, then $\textsc{PoA}_{\gamma} \leq 3 + 2 \gamma$.
\item if $\mathcal{V} \subseteq \xos$ and $\add \subseteq \mathcal{B} \subseteq \gs$, then $\textsc{PoA}_{\gamma} \leq 2 + \gamma$.
\end{itemize}
The bounds also hold for Bayes-Nash equilibria as long as the distribution from which valuations are sampled is independent across agents.
\end{theorem}

\subsection{Declared Efficiency Maximizers for gross substitute bidders}

The following lemma that will be useful in the proof of Theorem \ref{thm:gs_poa}. Recall the notation $W^{\vc{b}}(\vc{x}) = \max\{\sum_j b_j(\vc{x}_j); \sum_j \vc{x}_j \leq \vc{x}\}$ and $W^{\vc{b}_{-i}}(\vc{x}) = \max\{\sum_{j \neq i} b_j(\vc{x}_j); \sum_j \vc{x}_j \leq \vc{x}\}$.

\begin{lemma}\label{lemma:sum_welfare_gs}
If $b_1, \hdots, b_n \in \gs$, then for any partition $\vc{x}_i$ of the items,
i.e. $\sum_i \vc{x}_i = \one$, it holds that $\sum_i W^{\vc{b}_{-i}}(\vc{x}_i
\vert \one - \vc{x}_i) \leq W^{\vc{b}}(\one)$
\end{lemma}

\begin{proof}
A key observation is that the $\gs$ is closed under the OR operator, so $W^{\vc{b}}$ and $W^{\vc{b}_{-i}}$ are gross substitutes and therefore submodular. Using this fact, we first show that $W^{\vc{b}_{-i}}(\vc{x}_i \vert \one - \vc{x}_i) \leq
W^{\vc{b}}(\vc{x}_i \vert \one - \vc{x}_i)$. In order to see that, let $\vc{y}_i$ be player $i$ allocation in an  optimal partition of $\one-\vc{x}_i$ according to $\vc{b}$, i.e, $W^{\vc{b}}(\one - \vc{x}_i) = b_i(\vc{y}_i) + W^{\vc{b}_{-i}}(\one-\vc{x}_i-\vc{y}_i)$. Notice that $W^{\vc{b}}(\one) \geq b_i(\vc{y}_i) + W^{\vc{b}_{-i}}(\one - \vc{y}_i)$. Therefore, we have:
$$\begin{aligned}W^{\vc{b}}(\vc{x}_i \vert \one - \vc{x}_i) &= W^{\vc{b}}(\one) - W^{\vc{b}}(\one - \vc{x}_i) \geq [b_i(\vc{y}_i) + W^{\vc{b}_{-i}}(\one - \vc{y}_i)] - [b_i(\vc{y}_i) + W^{\vc{b}_{-i}}(\one-\vc{x}_i-\vc{y}_i)]\\
&=W^{\vc{b}_{-i}}(\vc{x}_i \vert \one - \vc{y}_i - \vc{x}_i) \geq W^{\vc{b}_{-i}}(\vc{x}_i \vert \one - \vc{x}_i)
\end{aligned}$$
where the last step follows from the submodularity of $W^{\vc{b}_{-i}}$. Now, we can apply submodularity of $W^{\vc{b}}$ and a telescopic sum:
$$\textstyle \sum_i W^{\vc{b}_{-i}}(\vc{x}_i \vert \one - \vc{x}_i) \leq \sum_i W^{\vc{b}}(\vc{x}_i \vert \one - \vc{x}_i) \leq \sum_i W^{\vc{b}}(\vc{x}_i \vert \sum_{j<i} \vc{x}_j) = W^{\vc{b}}(\one)$$
\end{proof}

\begin{proofof}{Theorem \ref{thm:gs_poa}}
Fix a declared efficiency maximizer mechanism given by $\vc{x}_i: \gs^n
\rightarrow \{0,1\}^m$ and $\pi_i:\gs^n \rightarrow \R_+$ for $i=1,...,n$ and a profile of $\gs$ valuations $\vc{v} = (v_1, \hdots, v_n)$ and let $\vc{b} = (b_1, \hdots, b_n) \in \gs^n$ be a pure Nash equilibrium of this mechanism for the type profile defined by $\vc{v}$. Let $\vc{x}_1, \hdots, \vc{x}_n$ be the
allocation and $(\pi_1, \hdots, \pi_n)$ be the payments in this equilibrium. Also, let
$\vc{x}_1^*, \hdots, \vc{x}_n^*$ be an allocation maximizing $\sum_i
v_i(\vc{x}_i^*)$.

Consider the utility of player $i$ by deviating to $b'_i = \frac{1}{2} v_i$ and
let $\vc{x}'_i$ be his allocation under the bids $(b'_i, \vc{b}_{-i})$. By
definition $\vc{x}'_i$ maximizes $\frac{1}{2} v_i(\vc{x}'_i) +
W^{\vc{b}_{-i}}(\one - \vc{x}'_i)$. In particular: $\frac{1}{2} v_i(\vc{x}'_i) + W^{\vc{b}_{-i}}(\one -
\vc{x}'_i) \geq \frac{1}{2} v_i(\vc{x}^*_i) + W^{\vc{b}_{-i}}(\one -
\vc{x}^*_i)$. Therefore we have:
$$\begin{aligned}
\comment{ v_i(\vc{x}_i) & \geq u_i(\vc{b}) \geq } u_i(b'_i, \vc{b}_{-i}) &= v_i(\vc{x}'_i) - \pi_i(b'_i,
\vc{b}_{-i}) \geq v_i(\vc{x}'_i) - \textstyle\frac{1}{2} v_i(\vc{x}'_i) =
\textstyle\frac{1}{2} v_i(\vc{x}'_i)\\
& \geq \textstyle\frac{1}{2} v_i(\vc{x}^*_i) + W^{\vc{b}_{-i}}(\one -
\vc{x}_i^*) - W^{\vc{b}_{-i}}(\one - \vc{x}'_i) \geq \\ & \geq \textstyle\frac{1}{2} v_i(\vc{x}^*_i) + W^{\vc{b}_{-i}}(\one -
\vc{x}_i^*) - W^{\vc{b}_{-i}}(\one) = \textstyle\frac{1}{2}
v_i(\vc{x}^*_i) - W^{\vc{b}_{-i}}(\vc{x}_i^* \vert \one - \vc{x}_i^*)
\end{aligned}$$

Summing over all agents $i$ and using Lemma \ref{lemma:sum_welfare_gs}, we get:
$$\sum_i u_i(b'_i, \vc{b}_{-i}) \geq \frac{1}{2}
v_i(\vc{x}^*_i) - \sum_i b_i(\vc{x}_i)$$
Now, we observe that since this is a Nash equilibrium, $u_i(b'_i, \vc{b}_{-i}) \leq u_i(\vc{b}) \leq v_i(\vc{x}_i)$. Also, by Observation \ref{obs:overbidding} for exposure factor $\gamma$, $b_i(\vc{x}_i) \leq (1+\gamma) v_i(\vc{x}_i)$. Therefore:
$$2(2+\gamma) \sum_i v_i(\vc{x}_i) \geq \sum_i v_i(\vc{x}_i^*)$$
\end{proofof}

The proof of Theorem \ref{thm:gs_poa} can be adapted to the Bayesian setting using the quasi-smoothness technique in Lucier and Paes Leme \cite{LPL11}. We reproduce it here for completeness:

\begin{proofof}{Theorem \ref{thm:gs_bpoa}}
Consider a Bayes-Nash equilibrium $\vc{b} = (b_1, \hdots, b_n)$ where $b_i: \mathcal{V} \rightarrow \mathcal{B}$. We know from the proof of Theorem \ref{thm:xos_poa} that each player has a deviation $b'_i(v_i) = \frac{1}{2} v_i$ that depends \emph{only on his valuation} such that for every realization of $\vc{v}$ :
$$\sum_i u_i(b'_i(v_i), \vc{b}_{-i}(\vc{v}_{-i}))   \geq \frac{1}{2} \sum_i v_i(\vc{x}_i^*(\vc{v})) - \sum_i b_i(\vc{x}_i(\vc{b})) $$
where $\vc{x}_i^*(\vc{v})$ is the allocation of $i$ in an optimal allocation with respect to the true valuations. Taking expectations, we obtain:
$$\sum_i \E_{\mathbf{D}} [u_i(b'_i(v_i), \vc{b}_{-i}(\vc{v}_{-i}))]   \geq \frac{1}{2} \sum_i \E_{\mathbf{D}}[v_i(\vc{x}_i^*(\vc{v}))] - \sum_i \E_{\mathbf{D}} [b_i(\vc{x}_i(\vc{b}))]$$
Since $\vc{b}$ is a Bayes-Nash equilibrium, we know that:
$$\E_{\mathbf{D}} [u_i(b'_i(v_i), \vc{b}_{-i}(\vc{v}_{-i}))] = \E[ \E [u_i(b'_i(v_i), \vc{b}_{-i}(\vc{v}_{-i})) \vert v_i]] \leq \E [\E [u_i(b_i(v_i), \vc{b}_{-i}(\vc{v}_{-i})) \vert v_i]] = \E[u_i(\vc{b}(\vc{v}))]$$
Using the previous line as well as Observation \ref{obs:overbidding}, we get:
$$(4+2\gamma) \E_{\mathbf{D}} [v_i(\vc{x}_i(\vc{b}(\vc{v})))] \geq \E_{\mathbf{D}} [v_i(\vc{x}_i^*(\vc{v}))] $$
\end{proofof}

\subsection{Declared Efficiency Maximizers for XOS bidders}

We note that the only point in the proof of Theorem \ref{thm:gs_poa} that we used that valuations are gross substituted was inside Lemma \ref{lemma:sum_welfare_gs} to argue that $W^{\vc{b}}$ is a submodular function. Even when $b_i$ are submodular for all $i$, $W^{\vc{b}}$  might fail to be submodular \cite{LehmannLehmannNisan}. To go around this problem, we prove a version of Lemma \ref{lemma:sum_welfare_gs} for the broader class of $\xos$ valuations.

\begin{lemma}\label{lemma:sum_welfare}
If $b_1, \hdots, b_n \in \xos$, then for any partition $\vc{x}_i$ of the items,
i.e. $\sum_i \vc{x}_i = \one$, it holds that $\sum_i W^{\vc{b}_{-i}}(\vc{x}_i
\vert \one - \vc{x}_i) \leq 2 W^{\vc{b}}(\one)$
\end{lemma}

\begin{proof}
Let $\vc{\hat{x}}_1, \hdots, \vc{\hat{x}}_n$ be a partition maximizing the
declared welfare, i.e., $\sum_i b_i(\vc{\hat{x}}_i) = W^{\vc{b}}(\one)$. Since
$b_i \in \xos$ there is a vector $\vc{w}_i \in \R_+^m$ such that
$b_i(\vc{\hat{x}}_i) = \vc{w}_i \cdot \vc{\hat{x}}_i$ and for any $\vc{y} \in
\{0,1\}^m$, $b_i(\vc{y}) \geq \vc{w}_i \cdot \vc{y}$. Therefore:
$$\sum_i W^{\vc{b}_{-i}}(\vc{x}_i \vert \one - \vc{x}_i) = \sum_i \left[
W^{\vc{b}_{-i}}(\one) - W^{\vc{b}_{-i}}(\one - \vc{x}_i) \right] \leq \sum_i
\left[ W^{\vc{b}_{-i}}(\one) - \sum_{k \neq i} b_k(\vc{\hat{x}}_k \cap (\one -
\vc{x}_i)) \right]$$
where the last inequality follows from the fact that $W^{\vc{b}_{-i}}(\one -
\vc{x}_i)$ is the value of the optimal allocation of $\one - \vc{x}_i$ to agents
$k \neq i$ and $\sum_{k \neq i} b_k(\vc{\hat{x}}_k \cap (\one - \vc{x}_i))$ is
the value of a particular allocation. Now, we can bound the value of
$b_k(\vc{\hat{x}}_k \cap (\one - \vc{x}_i))$ using the $\vc{w}_k$ vectors:
$$\sum_i \sum_{k \neq i} b_k(\vc{\hat{x}}_k \cap (\one - \vc{x}_i)) \geq \sum_i
\sum_{k \neq i} \vc{w}_k \cdot (\vc{\hat{x}}_k \cap (\one - \vc{x}_i)) \geq
(n-2) \sum_i \vc{w}_i \cdot \vc{\hat{x}}_i = (n-2) \sum_i b_i(\vc{\hat{x}}_i) $$
since for each $k$ there are $n-1$ terms of type $\vc{w}_k \cdot
(\vc{\hat{x}}_k \cap (\one - \vc{x}_i))$ and each term in $\vc{w}_k \cdot
\vc{\hat{x}}_k$ appears in all but one of them, since each term corresponds to
an item $j \in \vc{\hat{x}}_k$ and appears in all terms except the one for which
$j \in \vc{x}_i$. Therefore, we have:
$$\sum_i W^{\vc{b}_{-i}}(\vc{x}_i \vert \one - \vc{x}_i) \leq \sum_i
W^{\vc{b}_{-i}}(\one) - (n-2) W^{\vc{b}}(\one) \leq 2 \cdot W^{\vc{b}}(\one) $$
\end{proof}

\subsection{Efficiency of Equilibria for the VCG mechanism}

We next show how to improve our bounds on the efficiency of equilibria, for
the particular case of the VCG mechanism.

%For the special case of the VCG mechanism, we further improve the bound,
%matching the lower bound in Example \ref{example:2_bound_unit_demand}:

\begin{proofof}{Theorem \ref{thm:vcg_poa}}
Let $\vc{v} = (v_1, \hdots, v_n)$ be a valuation profile and $\vc{b}$ a Nash equilibrium of the VCG mechanism. Let also $\vc{x}_1^*, \hdots, \vc{x}_n^*$ be an optimal allocation with respect to the true valuations. We consider a deviation where player $i$ bids his true value instead:
$$\begin{aligned}u_i(\vc{b}) \geq u_i(v_i, \vc{b}_{-i}) &= v_i(\vc{x}_i(v_i, \vc{b}_{-i})) - \vc{W}^{\vc{b}_{-i}}(\vc{x}_i(v_i, \vc{b}_{-i}) \vert \one - \vc{x}_i(v_i, \vc{b}_{-i})) = \\ &= v_i(\vc{x}_i(v_i, \vc{b}_{-i})) + \vc{W}^{\vc{b}_{-i}}( \one - \vc{x}_i(v_i, \vc{b}_{-i})) - \vc{W}^{\vc{b}_{-i}}(\one) \\ & \geq  v_i(\vc{x}_i^*) + \vc{W}^{\vc{b}_{-i}}( \one - \vc{x}_i^*) - \vc{W}^{\vc{b}_{-i}}(\one) = v_i(\vc{x}_i^*) - \vc{W}^{\vc{b}_{-i}}( \vc{x}_i^* \vert \one - \vc{x}_i^*) \end{aligned}$$
Summing for all agents $i$ and applying Lemma \ref{lemma:sum_welfare} in case of $\xos$ valuations and Lemma \ref{lemma:sum_welfare_gs} in case of $\gs$ valuations, we get the desired bound. The extension to the Bayesian case follows the exact same arguments as in the proof of Theorem \ref{thm:gs_bpoa}.
\end{proofof}

\subsection{Restricted Bidding Languages}

Finally, we now analyze equilibria of welfare-maximizing mechanisms when the bidding language is not necessarily identical to the agents' type space, but is rather assumed only to be a subset of the type space that includes the set of additive valuation functions.

\begin{proofof}{Theorem \ref{thm:bidding_language}}
The proofs of Theorems \ref{thm:xos_poa} and \ref{thm:gs_poa} required that a player with valuation $v_i \in \mathcal{V}$ was able to bid $b_i = \frac{1}{2} v_i \in \mathcal{B}$ . Looking closer, one realizes that in fact, there only needs to exist a bid $b'_i \in \mathcal{B}$ with the following property: $b'_i(\vc{x}_i) \leq \frac{1}{2} b'_i(\vc{x}_i)$ for all $\vc{x}_i \in \{0,1\}^m$ and $b'_i(\vc{x}_i^*) = \frac{1}{2} b'_i(\vc{x}_i^*)$, where $(\vc{x}_1^*, \hdots, \vc{x}_n^*)$ is an optimal allocation with respect to bids. A simple observation is that if $v_i \in \xos$, then there is such bid $b'_i \in \add$. Notice that $v_i(\vc{x}) = \max_{j \in I} \vc{w}_j \cdot \vc{x}$, so for some $j \in I$, $v_i(\vc{x}_i^*) = \vc{w}_j \cdot \vc{x}_i^*$, so simply take $b'_i(\vc{x}) = \frac{1}{2} \vc{w}_j \cdot \vc{x}_i$.

For VCG instead of declared welfare maximizers, one can use the same argument without the half factor, i.e., one needs a deviating bid such that $b'_i(\vc{x}_i) \leq b'_i(\vc{x}_i)$ for all $\vc{x}_i \in \{0,1\}^m$ and $b'_i(\vc{x}_i^*) = b'_i(\vc{x}_i^*)$. One can simply take $b'_i(\vc{x}) = \vc{w}_j \cdot \vc{x}_i$.
\end{proofof}

The arguments used in the proof of Theorem \ref{thm:gs_bpoa} to extend Theorem \ref{thm:gs_poa} to the Bayesian case rely on the deviation $b'_i = \frac{1}{2} v_i$ depending only on the type of $i$. In the above proof, however, the deviation $b'_i(\vc{x}_i) = \vc{w}_i \cdot \vc{x}_i$ depends not only on $v_i$ but also on $\vc{x}_i^*$, which is a function of the entire valuation profile $\vc{v}$. Using the technique recently introduced by Syrgkanis and Tardos \cite{Syrgkanis13}, however, one can obtain Price of Anarchy bounds in the Bayesian setting from "smoothness-type" proofs for the case where the distribution over valuations is independent across agents.

\section{Conclusion}

We investigate the efficiency of the Walrasian mechanism when agents
strategically report their demands. It is known from Jackson and Manelli \cite{Jackson_Manelli}, Roberts and Postelwaite \cite{Roberts_Postlewaite} and recently Azevedo and Budish \cite{azevedo_budish} that in the limit as the market grows large, the players have little incentive to misreport their true demand and therefore the equilibrium approaches the market outcome with respect to the true preferences. In this paper, we analyze the small market regime, which models situations where the market is very specialized with few players or the market has a few major players whose transactions considerably affect the prices. Such situations are not uncommon in niches of the financial market, for example.

Without any assumptions on size of the market or on the distributions under which valuations are drawn, we show a bound on the efficiency of the market, measured in terms of the ratio between the optimal welfare and the welfare of the worst Nash equilibria. We show, however, that the efficiency crucially depend on a parameter that measures the amount of ''risk'' players are willing to expose themselves, which we call the exposure factor.

Our results are extended to the broader class of declared welfare maximizer mechanisms, which are mechanism that elicit bids from the agents in form or valuations over the items and allocate according to the optimal allocation in the declared market and charges prices that are at most the bids. Besides the Walrasian mechanism, this includes also the VGC mechanism and the pay-your-bid mechanism.

For mechanisms such as VCG and pay-your-bid, it is clearly  dominated to employ strategies with exposure factor $\gamma > 0$. So under elimination of weakly dominated strategies, our bounds hold with $\gamma = 0$ for such mechanisms. We also show that the same is not true for the Walrasian mechanism. We leave as an open question if there is some $\gamma > 0$ for which strategies with exposure factor $\gamma$ are dominated.

Certain classes of valuations such as submodular or $\xos$ are so rich that don't allow for computationally efficient methods to reach the optimal allocation. We discuss one option to get around this problem in the paper: restrict the class of valuations that agents can submit as bids. We show that our results still holds if the restricted bidding language contains at lease additive valuations. An alternative solution would be to compute an approximately optimal allocation, say using the procedures of Lehmann, Lehmann and Nisan \cite{LehmannLehmannNisan} or Fu, Kleinberg and Lavi \cite{FKL12}. We believe an interesting open problem arising from this work is how to extend our results to \emph{approximate} declared welfare maximizers.

\bibliographystyle{abbrv}
\bibliography{sigproc}  % sigproc.bib is the name of the Bibliography in this case
% ACM needs 'a single self-contained file'!
%\balancecolumns
\appendix
\section{Overbidding}\label{appendix:overbidding}

In this appendix we present an example in which a bidder has a strict preference for declaring more than
his true value on a set of goods, in the English Walrasian mechanism.  
This establishes that such ``overbidding'' strategies are not weakly dominated in this mechanism.
Our example will consist of $3$ players and $3$ items. 
When two of the players honestly report their values, the third player will be incentivized
to overbid on one set. He does so to artificially inflate his welfare in
various allocations, which reduces the price he must pay for different goods.

The valuations of the three players are as follows:
\begin{itemize}
\item $v_1(\vc{x}) = 4 x_1 + 2 \max\{x_2, x_3\}$
\item $v_2(\vc{x}) = 2 \max\{x_1, x_2\}$
\item $v_1(\vc{x}) = x_3$
\end{itemize}

The optimal allocation is $\vc{x}_1 = \one_{\{1,3\}}$, $\vc{x}_2 =
\one_{\{2\}}$, and $\vc{x}_3 = \vc{0}$,
for a social welfare of $8$.  Furthermore, the English Walrasian mechanism sets prices
$p_j = 1$ for each $j \in [3]$; to see this, recall that the minimum Walrasian price vector is given by
$\underline{p}_j = W^{\vc{v}}(\one_j \vert \one)$, and note that under the addition of an extra copy of any
item the optimal allocation would generate social welfare $9$.  Thus, under a profile of truthful reports,
the utility of agent $1$ is $4$.

Now, consider a deviation in which player $A$ declares valuation $b_1'$, given
by
$$b_1'(\vc{x}) = 4 \cdot x_1 + \max\{2 x_2, 3 x_3\}. $$
This declaration overbids on sets $\{1,3\}$ and $\{3\}$.  Under valuation
profile $(b_1', \vc{v}_{-1})$, the optimal allocation
remains unchanged: it is $\vc{x}_1 = \one_{\{1,3\}}$, $\vc{x}_2 = \one_{\{2\}}$, and $\vc{x}_3 = \vc{0}$.
However, the Walrasian prices under this declaration profile are given by 
$p_1 = 0, p_2 = 0, p_3 = 1$, since neither an extra copy of item $1$ nor $2$ would lead to an optimal
allocation with increased \emph{reported} welfare.  Thus, under declared profile
$(b_1', \vc{v}_{-1})$, the
utility of agent $1$ is $6$, which is optimal over the space of possible declarations for agent $1$.
We conclude that declaration $b_1'$ is not weakly dominated for agent $1$ when
his true valuation is
$v_1$.  In particular, not all ``overbidding'' strategies are weakly dominated.

%Appendix A
\fullversion{}{}

%\balancecolumns
% That's all folks!
\end{document}